\documentclass{llncs}
\author{Kenya Ueno}
\authorrunning{Kenya Ueno}
\tocauthor{Kenya Ueno(Kyoto University)}
\institute{The Hakubi Center for Advanced Research, and\\
Graduate School of Informatics\\
Kyoto University\\
\email{kenya@i.kyoto-u.ac.jp}}
\titlerunning{}
\usepackage{amsmath,amssymb,algorithmic,algorithm}

\title{Exact Algorithms for 0-1 Integer Programs\\with Linear Equality Constraints}

\begin{document}
\maketitle
\begin{abstract}
In this paper, we show $O(1.415^n)$-time and $O(1.190^n)$-space exact algorithms for 0-1 integer programs where constraints are linear equalities and coefficients are arbitrary real numbers. Our algorithms are quadratically faster than exhaustive search and almost quadratically faster than an algorithm for an inequality version of the problem by Impagliazzo, Lovett, Paturi and Schneider (arXiv:1401.5512), which motivated our work. Rather than improving the time and space complexity, we advance to a simple direction as inclusion of many NP-hard problems in terms of exact exponential algorithms. Specifically, we extend our algorithms to linear optimization problems.
\end{abstract}

\section{Introduction}
\paragraph{The Feasibility Problem:}
The existence of integer solutions for a certain system of equations has been discussed 
as one of the fundamental problems in the theory of computation. 
A prominent example is the Hilbert 10th problem on Diophantine equations~\cite{Matiyasevich93}.

In this paper, we study the feasibility problem of 0-1 integer programs whose constraints are only linear equalities as follows:
\begin{problem}[Feasibility of 0-1 Integer Programs with Linear Equalities] \ 
\label{feasibility}
\begin{center}
Find $x \in \{0,1\}^n$ which satisfies a given set of linear equalities $A x = b$.
\end{center}
\end{problem}

We give an exact algorithm running in $O(1.415^n)$-time and $O(1.190^n)$-space, which achieves a quadratic speedup compared to exhaustive search running in $O(2^n)$-time.
Our algorithm can store the data of all the feasible solutions in $O(1.415^n)$-space, even if the number of solutions is more than $O(1.415^n)$.

As a similar problem, which achieves quadratic speedup,
there is a quantum algorithm known as Grover's algorithm for unstructured database search problems~\cite{Grover97}.
It gives a correct answer with high probability, but our algorithm do not use randomness and always gives a correct answer.
Recently, probabilistic polynomial algorithms solving a system of linear equations has been discussed by Raghavendra~\cite{Raghavendra12} and Fliege~\cite{Fliege12}.
If we eliminate the 0-1 constraints, we can give a polynomial time algorithm by the Gaussian elimination.

\paragraph{The Optimization Problem:}
Then, we extend our algorithm for the following standard optimization problem running in $O(1.415^n)$-time and $O(1.190^n)$-space:
\begin{problem}[Optimization of 0-1 Integer Programs with Linear Equalities]
\label{optimization}
\begin{eqnarray*}
\label{IP_primal}
\begin{array}{ll}
\min & \displaystyle c^T x \\
s.t. & \displaystyle A x = b,\\
& x \in \{0,1\}^n. 
\end{array}
\end{eqnarray*}
\end{problem}

We know that there are many sophisticated ideas (e.g., the branch-and-bound method and the cutting-plane method) improving algorithms 
and implementations for computing 0-1 integer programs~\cite{GG11,Wolsey98}. 
However, we don't know any improvements of worst-case time complexity for such a general setting in which elements of $A$ and $b$ can be arbitrary real numbers.

\paragraph{Exact Algorithms for NP-hard Problems:}
Since there are no polynomial time algorithms for NP-hard problems unless P=NP,
many researchers have studied exact exponential time algorithms which are faster than exhaustive search for NP-hard problems~\cite{FK10,IP01,IPZ01,Woeginger03}.
Integer programs include many NP-hard problems as special cases.
For instance, the subset sum problem is a special case of Problem~\ref{feasibility} in which the number of constraints is exactly one.

Among several such problems whose exact algorithms have been studied,
some problems (e.g., the subset sum problem~\cite{HS74,SS81}) have the same time complexity as Problem~\ref{feasibility}, and
some other problems (e.g., the exact satisfiability problem~\cite{BMS05} and the exact hitting set problem~\cite{DP02}) have algorithms faster than $O(1.415^n)$-time.
In particular, the exact satisfiability problem, which is also a special case of Problem~\ref{feasibility}, has been intensively studied~\cite{BMS05,DP02}.

On the other hand, it seems to be difficult to improve the time complexity of our algorithms due to a similar reason of NP-hardness.
In other words, if we can improve our algorithms, then we simultaneously improve the time complexity of exact algorithms
for many NP-hard problems which can be reduced to Problem~\ref{feasibility}.

\paragraph{Circuit Lower Bounds from Moderately Exponential Algorithms:}
Very recently, Impagliazzo, Lovett, Paturi and Schneider~\cite{ILPS14} studied the feasibility problem for the inequality version of 0-1 integer programs stated as follows:
\begin{problem}[Feasibility of 0-1 Integer Programs with Linear Inequalities] \ 
\label{inequality}
\begin{center}
Find $x \in \{0,1\}^n$ which satisfies a given set of linear inequalities $A x \geq b$.
\end{center}
\end{problem}

Impagliazzo, Lovett, Paturi and Schneider~\cite{ILPS14} gave an algorithm solving Problem~\ref{inequality} in $O(2^{(1 - \mathrm{poly}(1/d))n})$-time where $dn$ is the number of constraints.
It improves an algorithm for Problem~\ref{inequality} by Impagliazzo, Paturi and Schneider~\cite{IPS13}, which is faster than 
$O(2^n)$-time only when the number of inequalities is smaller than $0.136 n$.
These results are motivated from the challenge initiated by Williams~\cite{Williams13}
for proving lower bounds for certain circuit models.
In this context, it is important to give only a modest improvement of the exponential factor from the $O(2^n)$-time exhaustive search.

\paragraph{Our Algorithms:}
Our algorithms are built on a simple combination of basic techniques on exact algorithms for NP-hard problems. 
In particular, we use a classic technique called the $k$-table method studied in \cite{HS74,SS81}.
This method splits $n$-variables into the $k$ sets of $n/k$-variables, 
and lists all possible $2^{n/k}$-assignments for each set.
This preprocessing enables us to give algorithms which run faster than $O(2^n)$-time for certain problems.

On the other hand, it was unclear how we construct the $k$-table 
for exact algorithms to compute 0-1 integer programs with linear equalities.
The ideas introduced in the two papers~\cite{ILPS14,IPS13} give us inspiration 
to overcome technical problems including analysis to bound the time complexity.
In this paper, we introduce a notion of the vector equality problem, a variation of 
the vector domination problem studied by Impagliazzo, Paturi and Schneider~\cite{IPS13}, to construct the 2-table for our problems.
In Section~\ref{vector_section}, we show two algorithms solving the vector equality problem and give analysis of its time complexity.

In Section~\ref{IP_section}, we describe how we compute the feasibility problem and the optimization problem for 0-1 integer programs with linear equalities by reducing them to the vector equality problem.
A novel point here is an extension of the feasibility problem to the optimization problem by using the 2-table method.
This is achieved by post-processing after solving the feasibility problem with extra storage of the objective function.
There are no extra blow-up of the exponential time complexity.
In Section~\ref{space_section}, we improve the space complexity of our algorithms to $O(1.190^n)$ following the idea of Shroeppel and Shamir~\cite{SS81}.

\section{Notation and Definition}
Throughout the paper, we use the following notations.
We denote $m \times n$ constant matrices by $A$, and $i,j$-th element of a matrix $A$ by $A_{i,j}$.
We use $b$ and $c$ as constant vectors. $c^T$ is the transpose of $c$.
We also use $x$, $u$ and $v$ as variable vectors.
We denote $j$-th element of a vector $x$ by $x_j$.
The same notation applies for other constant and variable vectors.

The function $\mathrm{poly}(n)$ is some polynomial for $n$.
Following the convention in the theory of exact algorithms, we measure the time complexity by the function of $n$, which is the number of variables.
We assume $m \in O(\mathrm{poly}(n))$ since otherwise the input size is super-polynomial to $n$.

In this paper, we will give algorithms for 0-1 integer programs with linear equality constraints by reducing them to following problem.
\begin{definition}[Vector Equality]
Two vectors $u=(u_1,u_2,\cdots,u_m)$ and $v=(v_1,v_2,\cdots,v_m)$ are equal (i.e., $u=v$) if and only if $u_i=v_i$ for any $i~(0 \leq i \leq m)$.
Given two sets of $m$-dimensional vectors $U$ and $V$, the {\it vector equality problem} is 
a problem to output information (e.g., a list of subsets of $U$ and $V$) of all the pairs of two vectors $u \in U$ and $v \in V$ such that $u = v$.
\end{definition}
Note that elements of vectors can be real numbers. 
Therefore, we cannot combine a set of elements into one element unlike the case of integers whose absolute values are bounded~\cite{CD99}.

We will use the lexicographical order to compare two $m$-dimensional vectors.
\begin{definition}[Lexicographical Order]
A vector $u = (u_1,u_2,\cdots,u_m)$ is larger than another vector $v = (v_1,v_2,\cdots,v_m)$ (i.e., $u > v$) if and only if there is an index $l~(0 \leq l \leq m)$ such that
$u_l > v_l$ and $u_i=v_i$ for any $i~(0 \leq i < l)$.
\end{definition}

\section{Algorithms for the Vector Equality Problem}
\label{vector_section}
\subsection{Overview and Comparison of Two Algorithms}
In this section, we present two algorithms (Algorithm~\ref{sort_algo} and Algorithm~\ref{recursive_algo}) 
to solve the vector equality problem efficiently.
Algorithm~\ref{sort_algo} is simple and uses a sort routine by lexicographical order between two $m$-dimensional vectors, 
while Algorithm~\ref{recursive_algo} is recursive and uses the idea of measure-and-conquer~\cite{FK10}.

Although its theoretical time complexity is the same as the other one, 
Algorithm~\ref{recursive_algo} has a practical merit when $m$ (the number of constraints) is large.
Moreover, we can see the quadratic difference between 
the equality and inequality versions of integer programs by looking at the recursive algorithm
compared to the algorithms for Problem~\ref{inequality}~\cite{ILPS14,IPS13}.

To analyze the time complexity of Algorithm~\ref{recursive_algo}, 
we need to incorporate an idea using the weighted median, 
which is introduced very recently by Impagliazzo, Lovett, Paturi and Schneider~\cite{ILPS14}.
Furthermore, we complete our analysis of the time complexity by setting a suitable choice of complexity measure, which is a novel point of this paper, for the search space in the 2-table method.

Algorithm~\ref{recursive_algo} can be regarded as a vector version of the quicksort algorithm 
with the weighted median as a pivot. 
Actually, it is not necessary to use the weighted median in practice to select the pivot.
A heuristical choice of the pivot may be faster in many cases, 
while there are no certificates of its worst-case time complexity.
After the completion of the first draft including Algorithm~\ref{recursive_algo},
we noticed that it can be simplified as Algorithm~\ref{sort_algo}.
However, we consider Algorithm~\ref{recursive_algo} is still beneficial to 
present due to several reasons as mentioned above.

\subsection{A Simple Algorithm by Sorting}
\begin{algorithm}[ht]
\caption{SortVectorEquality$(U,V)$}
\begin{algorithmic}
\REQUIRE Two sets of $m$-dimensional vectors
\ENSURE A list of two sets of $m$-dimensional vectors
\STATE Sort $U$ and $V$ in the ascending lexicographical order, respectively.\\
($u^k$ and $v^k$ denote the $k$-th vectors in $U$ and $V$, respectively)
\STATE Initialize two indices $\alpha = 1$ and $\beta = 1$ for $U$ and $V$.
\WHILE{$\alpha \leq |U|$ \AND $\beta \leq |V|$}
\IF{$u^\alpha > v^\beta$}
\STATE Increment $\beta$.
\ELSIF{$u^\alpha < v^\beta$}
\STATE Increment $\alpha$.
\ELSE
\STATE Set $\alpha' := \alpha$, $\beta' := \beta$ and $w := u^\alpha$ ($:= v^\beta$).
\WHILE{$w = u^\alpha$}
\STATE Increment $\alpha$.
\ENDWHILE
\WHILE{$w = v^\beta$}
\STATE Increment $\beta$.
\ENDWHILE
\STATE Output $(\alpha',\alpha-1)$ and $(\beta',\beta-1)$ as representation of two subsets of $U$ and $V$.
\ENDIF
\ENDWHILE
\end{algorithmic}
\label{sort_algo}
\end{algorithm}

We describe a simple algorithm by sorting for solving the vector equality problem in Algorithm~\ref{sort_algo}.
Since the sort of $m$-dimensional vectors can be done in $O(m N \log N)$-time, we can also run Algorithm~\ref{sort_algo} in $O(m N \log N)$-time.
\begin{lemma}
\label{vector_lemma1}
The vector equality problem can be computed in $O(m N \log N)$-time where $|U| = |V| = N$ by Algorithm~\ref{sort_algo}.
\end{lemma}

Our algorithm can enumerate all the possible solutions. 
It may sound strange that we can store the data of all possible solutions within $O(m N \log N)$-space, even if the number of all possible solutions is $\omega(m N \log N)$.
This is just because we store the data as a collection of two sets of elements.

If the number of solutions is bounded by $O(m N \log N)$, then the time complexity to enumerate all the solutions is also $O(m N \log N)$.
If the number of solutions is $\omega(m N \log N)$, then the time complexity depends on the number of possible solutions.

\subsection{A Recursive Algorithm by Measure-and-Conquer}
Another way to solve the vector equality problem is to use 
a notion of the weighted median to bound the time complexity of our algorithms
in the way of measure-and-conquer.

\begin{definition}[Weighted Median]
The weighted median for a set of weighted numbers is a number such that both the total weight of numbers smaller than the weighted median 
and the total weight of numbers larger than the weighted median are at most half of the total weight of all the numbers.
\end{definition}
\begin{algorithm}[ht]
\caption{RecursiveVectorEquality$(U,V,i,m)$} \label{recursive_algo}
\begin{algorithmic}
\REQUIRE Two sets of $m$-dimensional vectors and an index $i$ and the dimension $m$.
\ENSURE A list of two sets of $m$-dimensional vectors.
\IF{$U = \emptyset$ or $V = \emptyset$}
\RETURN an empty list
\ELSIF{$i > m$}
\RETURN a singleton list of $(U,V)$
\ELSE
\STATE 
\begin{enumerate}
\item Find the weighted median $k$ of the $i$-th coordinates of $U \cup V$ with weight $|V|$ and $|U|$ for each element in $U$ and  $V$, respectively.
\item Partition $U$ into three sets:
\begin{enumerate}
\item $U^+ = \{ u \mid u_i > k \}$,
\item $U^= = \{ u \mid u_i = k \}$,
\item $U^- = \{ u \mid u_i < k \}$.
\end{enumerate}
\item
Partition $V$ into three sets: 
\begin{enumerate}
\item $V^+ = \{ u \mid v_i > k \}$,
\item $V^= = \{ u \mid v_i = k \}$,
\item $V^- = \{ u \mid v_i < k \}$.
\end{enumerate}
\item
Solve the following three subproblems:
\begin{enumerate}
\item L1 = VectorEquality$(U^+,V^+,i,m)$
\item L2 = VectorEquality$(U^=,V^=,i+1,m)$
\item L3 = VectorEquality$(U^-,V^-,i,m)$
\end{enumerate}
\end{enumerate}
\RETURN the concatenation of the three lists L1, L2, and L3
\ENDIF
\end{algorithmic}
\end{algorithm}

Then, we consider a recursive algorithm (Algorithm~\ref{recursive_algo}) computing the vector equality problem.
Following a linear time algorithm for the unweighted median problem~\cite{Blum72},
we can give a linear time algorithm for the weighted median problem~\cite{BO83}, which is also indicated in \cite{ILPS14}.
\begin{lemma}[\cite{BO83,ILPS14}]
The weighted median of $N$ numbers can be computed in $O(N)$-time. 
\end{lemma}

We analyze the time complexity of Algorithm~\ref{recursive_algo} for the vector equality problem in the following lemma.
\begin{lemma}
\label{vector_lemma2}
The vector equality problem can be computed in $O(m N \log N)$-time where $|U| = |V| = N$ by starting Algorithm~\ref{recursive_algo} at RecursiveVectorEquality$(U,V,1,m)$.
\end{lemma}
\begin{proof}
In Algorithm~\ref{recursive_algo}, we find the weighted median $k$ of the $i$-th coordinates of $U \cup V$ where all the elements in $U$ and $V$ have weight $|V|$ and $|U|$, respectively.
Then, we partition each of $U$ and $V$ into three sets, respectively.

Two vectors $u \in U$ and $v \in V$ can be equal in at most one of the following three cases:
\begin{center}
(1)~ $u \in U^+$ and $v \in V^+$,\\
(2)~ $u \in U^=$ and $v \in V^=$,\\
(3)~ $u \in U^-$ and $v \in V^-$.\\
\end{center}
We solve smaller subproblems of the vector equality problem for the three cases. 
In particular, we decrease the dimension $m$ to $m-1$ in the case of (2). 

The rule of the partition immediately gives the following equation:
\begin{align*}
|V| \cdot (|U^+| + |U^=| + |U^-|) + |U| \cdot (|V^+| + |V^=| + |V^-|) 
=  |V| \cdot |U| + |U| \cdot |V|.
\end{align*}
Dividing it by $|U| \cdot |V|$, we have
\[
\frac{|U^+| + |U^=| + |U^-|}{|U|} + \frac{|V^+| + |V^=| + |V^-|}{|V|} = 2.
\]
For some constants $s$ and $t$ such that $0 \leq s \leq 1$ and $0 \leq t \leq 1$, we have
\begin{align*}
&\frac{|U^+|}{|U|} + \frac{|V^+|}{|V|} = 1 - s,\\
&\frac{|U^-|}{|U|} + \frac{|V^-|}{|V|} = 1 - t,\\
&\frac{|U^=|}{|U|} + \frac{|V^=|}{|V|} = s + t
\end{align*}
because we partitioned $U$ and $V$ at the weighted median.
Since $\alpha + \beta \geq 2\sqrt{\alpha \beta} $ for any $\alpha,\beta \geq 0$, we have
\begin{align*}
&\frac{|U^+|}{|U|} \cdot \frac{|V^+|}{|V|} \leq \frac{1}{4} \cdot (1 - s )^2,\\
&\frac{|U^-|}{|U|} \cdot \frac{|V^-|}{|V|} \leq \frac{1}{4} \cdot (1 - t )^2,\\
&\frac{|U^=|}{|U|} \cdot \frac{|V^=|}{|V|} \leq \frac{1}{4} \cdot (s + t )^2.
\end{align*}
Collecting these inequalities, we have
\begin{align*}
&|U^+| \cdot |V^+| \cdot 2^m + |U^-| \cdot |V^-| \cdot 2^m + |U^=| \cdot |V^=| \cdot 2^{m-1} \\
\leq &\frac{1}{4} \cdot \{ (1 - s )^2 \cdot 2^m +  (1 - t )^2 \cdot 2^m + (s + t )^2 \cdot 2^{m-1} \} \cdot |U| \cdot |V|\\
= &\frac{1}{4} \cdot \{ (1 - s )^2 +  (1 - t )^2  + \frac{1}{2} \cdot (s + t )^2 \} \cdot |U| \cdot |V| \cdot 2^m.
\end{align*}
It means that the search space $|U| \cdot |V| \cdot 2^m$ decreases by the factor of
\begin{align*}
f (s,t) &= \frac{1}{4} \cdot \{  (1 - s )^2 +  (1 - t )^2 + \frac{1}{2} (s + t )^2 \}\\
&= 0.5 - 0.5 s - 0.5 t + 0.375 {s}^2 + 0.375 {t}^2 + 0.25 s t
\end{align*}
at each recursion.

We can conclude $f (s,t) \leq \frac{1}{2}$ in the domain of $0 \leq s \leq 1$ and $0 \leq t \leq 1$ by the following argument.
By taking the partial derivatives, we have
\begin{align*}
\frac{\partial f (s,t)}{\partial s} = - 0.5 + 0.75 s + 0.25 t,\\
\frac{\partial f (s,t)}{\partial t} = - 0.5 + 0.25 s + 0.75 t.
\end{align*}
If $\frac{\partial f (s,t)}{\partial s} > 0$ (equivalently, $t > 2 - 3 s$), then the function $f(s,t)$ is monotonically increasing in the direction of $s$.
If $\frac{\partial f (s,t)}{\partial s} < 0$ (equivalently, $t < 2 - 3 s$), then the function $f(s,t)$ is monotonically decreasing in the direction of $s$.
The same thing applies for $t$ instead of $s$.

Therefore, we can verify that it is maximized at two edges $(s,t) = (0,0), (1,1)$ as $f (s,t) = 0.5$ 
and minimized at the middle point $(s,t) = (0.5, 0.5)$ as $f (s,t) = 0.25$.
Moreover, maximal points except the two edges are only two points $(s,t) = (0,1), (1,0)$ 
as $f (s,t) = 0.375$.

The recursions occur at most $\log_2 (|U| \cdot |V|  \cdot 2^m) \in O(m\log N)$ depth.
At each depth $d$ of the recursion, we need to solve at most $3^d$ ($< N$) subproblems of 
the vector equality problem, but the total number of elements is at most $2N$.
Therefore, we can solve the weighted median in linear time $O(|U|+|V|) = O(N)$ as a whole
at each depth of the recursion.

As a consequence, we conclude that the total time complexity of Algorithm~\ref{recursive_algo} is $O(m N \log N)$. 
\qed
\end{proof}

\section{Exact Algorithms for 0-1 Integer Programs}
\label{IP_section}
In this section, we give an exact algorithm for solving the feasibility and optimization problem of 0-1 integer programs with linear equality constraints
by reducing it to the vector equality problem described in the previous section.

\begin{theorem}
\label{feasibility_thm}
The feasibility and optimization problem of 0-1 integer programs with linear equalities (Problem~\ref{feasibility} and Problem~\ref{optimization}) can be computed in $O(m \cdot 2^{n/2} \mathrm{poly}(n))$-time.
\end{theorem}
\begin{proof}
We solve the feasibility problem of $Ax=b$ by reducing it to the vector equality problem.
First, we partition the set of variables $X = \{x_1, \cdots, x_n\}$ into two disjoint subsets $X_1$ and $X_2$.
Here, we assume the number of variables $n$ is even without loss of generality.

Let $\varphi(x_j)$ be assignments of $x_j$.
Then, we define vectors $u$ and $v$ by
\begin{align*}
&u_i = \sum_{x_j \in X_1} A_{ij} \cdot \varphi(x_j),
&v_i = b_i - \sum_{x_j \in X_2} A_{ij} \cdot \varphi(x_j).
\end{align*}
for each assignment of $X_1$ and $X_2$. Let $U$ and $V$ be two sets of $2^{n/2}$ such vectors $u$ and $v$, respectively.

Taking into account the linearity of the objective function $c^T x$ of Problem~\ref{optimization}, we can extend the algorithm 
for the feasibility problem to one for the optimization problem.
For this purpose, we additionally calculate weight
\begin{align*}
&w(u) =  \sum_{x_j \in X_1} c_j \cdot \varphi(x_j),
&w(v) =  \sum_{x_j \in X_2} c_j \cdot \varphi(x_j)
\end{align*}
for each of $u \in U$ and $v \in V$, respectively.

From the construction of $U$ and $V$, there is a 0,1-vector $x \in \{0,1\}^n$ satisfying $Ax = b$ if and only if 
there is a pair of two vectors $u \in U$ and $v \in V$ satisfying $u_i = v_i$ for all $i$ ($1 \leq i \leq m$).

We can solve the vector equality problem for $U$ and $V$ in $O(m N \log N)$-time by using algorithms in Section~\ref{vector_section}.
After the algorithms terminates, we can get a list of submatrices which contains information of all the possible solutions.
\[
(U^1,V^1), (U^2,V^2), \cdots, (U^k,V^k), \cdots , (U^l,V^l)
\]
There are at most $N = 2^{n/2}$ submatrices.
From the construction of the algorithms, each row and column of submatrices has no intersection.

Let $U^k \times V^k$ ($U^k \subseteq U$ and $V^k \subseteq V$) be one of such submatrices.
Then we would like to solve the following optimization problem for each $k$.
\begin{eqnarray*}
\begin{array}{ll}
\min & w(u) + w(v) \\
s.t. & u \in U^k \mbox{ and } v \in V^k. 
\end{array}
\end{eqnarray*}
From the linearity of $c^T$, $w(u)$ and $w(v)$ are independent.
Therefore , the above minimization problem is solvable separately for $u$ and $v$.
Hence, $O(|U^k| + |V^k|)$-time is sufficient to optimize.

We solve the same problem for each submatrices and take the minimum of all the problems.
The total time complexity is $O(m \cdot 2^{n/2} \mathrm{poly}(n))$.
\qed
\end{proof}

\section{Improved Space Complexity}
\label{space_section}
Shroeppel and Shamir~\cite{SS81} studied the $k$-table method, which is a generalization of the 2-table method, 
and showed an $O(2^{n/4})$-space exact algorithm for the subset sum problem (a special case of the Problem~\ref{feasibility}) by using the 4-table method.
Following the idea of Shroeppel and Shamir~\cite{SS81} using the priority queue, 
we can reduce the space complexity of our algorithms from $O(2^{n/2})$ to $O(2^{n/4})$ as in the following theorem.

\begin{theorem}
The feasibility and optimization problem of 0-1 integer programs with linear equality constraints
(Problem~\ref{feasibility} and Problem~\ref{optimization}) can be computed
in $O(m \cdot 2^{n/2} \mathrm{poly}(n))$-time and $O(m \cdot 2^{n/4} \mathrm{poly}(n))$-space.
\end{theorem}
\begin{proof}
We partition the set of variables $X = \{x_1, \cdots, x_n\}$ into four disjoint subsets $X_1$, $X_2$, $X_3$ and $X_4$.
Here, we assume the number of variables $n$ can be divided by 4 without loss of generality.

Let $\varphi(x_j)$ be an assignment of $x_j$.
Then, we define vectors $u$, $v$, $s$ and $t$ by
\begin{align*}
&u_i = \sum_{x_j \in X_1} A_{ij} \cdot \varphi(x_j),
&v_i = \sum_{x_j \in X_2} A_{ij} \cdot \varphi(x_j),\\
&s_i = - \sum_{x_j \in X_3} A_{ij} \cdot \varphi(x_j),
&t_i = b_i - \sum_{x_j \in X_4} A_{ij} \cdot \varphi(x_j).
\end{align*}
for each assignment of $X_1$, $X_2$, $X_3$ and $X_4$. 
Let $U$, $V$, $S$ and $T$ be four sets of $2^{n/4}$ such vectors $u$, $v$, $s$ and $t$, respectively.

We additionally calculate weight
\begin{align*}
&w(u) =  \sum_{x_j \in X_1} c_j \cdot \varphi(x_j),
&w(v) =  \sum_{x_j \in X_2} c_j \cdot \varphi(x_j),\\
&w(s) =  \sum_{x_j \in X_3} c_j \cdot \varphi(x_j),
&w(t) =  \sum_{x_j \in X_4} c_j \cdot \varphi(x_j).
\end{align*}
for each of $u \in U$, $v \in V$, $s \in S$ and $t \in T$, respectively.
For each vector, we can store data of its corresponding assignment and weight within $O(n)$-space.

From the construction of the four sets, there is a 0,1-vector $x \in \{0,1\}^n$ satisfying $Ax = b$ if and only if 
a quartet of four vectors $u \in U$, $v \in V$, $s \in S$ and $t \in T$ such that $u + v = s + t$.
We can search such quartets by Algorithm~\ref{pq_algo} in 
$O(m N^2 \log N)$-time and $O(m N)$-space where $|U| = |V| = |S| = |T| = N$.
In the algorithm, we use priority queues in which we can push and pop any element 
in the logarithmic time to the number of elements.

We can compute the minimum objective value and the corresponding assignment of the original 0-1 integer programs by Algorithm~\ref{pq_algo}
where the inputs are given as four sets of $m$-dimensional vectors with their assignments $\varphi$ and weights $w$.
The return value of $\infty$ means that the problem is infeasible.
If it is feasible, we can retrieve the corresponding assignment $\varphi$ of variables from the values of $\mathbf{SOL}$ in Algorithm~\ref{pq_algo}.
\qed
\end{proof}

\begin{corollary}
The feasibility and optimization problems of 0-1 integer programs with linear equality constraints (Problem~\ref{feasibility} and Problem~\ref{optimization}) 
can be computed in $O(1.415^n)$-time and $O(1.190^n)$-space.
\end{corollary}

\begin{algorithm}
\caption{VectorSumEquality$(U,V,S,T)$}
\begin{algorithmic}
\STATE Sort $U$, $V$ $S$ and $T$ in the ascending lexicographical order, respectively. \\
($u^k$, $v^k$, $s^k$, and $t^k$ denote the $k$-th vectors in $U$, $V$ $S$ and $T$, respectively)
\STATE Set $\mathbf{MIN} := \infty$ and initialize two priority queues $Q_1$ and $Q_2$ as empty sets.
\FOR{$k$ = 1 \TO $|V|$}
\STATE Push $(u^k, v^1)$ to the priority queue $Q_1$.
\ENDFOR
\FOR{$k$ = 1 \TO $|T|$}
\STATE Push $(s^k, t^1)$ to the priority queue $Q_2$.
\ENDFOR
\WHILE{Both of $Q_1$ and $Q_2$ are not empty}
\STATE Take the top elements $(u^\alpha,v^\beta)$ and $(s^\gamma,t^\delta)$ from $Q_1$ and $Q_2$, respectively.
\IF{$u^\alpha + v^\beta < s^\gamma + t^\delta$}
\STATE Pop $(u^\alpha,v^\beta)$. 
\IF{$\beta + 1 \leq |V|$}
\STATE Push  $(u^\alpha,v^{\beta+1})$.
\ENDIF
\ELSIF{$u^\alpha + v^\beta > s^\gamma + t^\delta$}
\STATE Pop $(s^\gamma,t^\delta)$. 
\IF{$\delta + 1 \leq |T|$}
\STATE Push $(s^\gamma,t^{\delta+1})$.
\ENDIF
\ELSE
\STATE $w := u^\alpha + v^\beta (:= s^\gamma + t^\delta)$;
\STATE $\mathbf{MIN_1} := \infty$;~ $\mathbf{MIN_2} := \infty$;
\WHILE{$Q_1 \neq \emptyset$ and $w = u^{\alpha'} + v^{\beta'}$ where $(u^{\alpha'},v^{\beta'})$ is the top element of $Q_1$}
\STATE Pop $(u^{\alpha'},v^{\beta'})$.
\IF{$\beta' + 1 \leq |V|$}
\STATE Push $(u^{\alpha'},v^{\beta'+1})$.
\ENDIF
\IF{$\mathbf{MIN_1} > w(u^{\alpha'}) + w(v^{\beta'})$}
\STATE $\mathbf{MIN_1} := w(u^{\alpha'}) + w(v^{\beta'})$;~ $\mathbf{SOL_1} := (\alpha',\beta')$;
\ENDIF
\ENDWHILE
\WHILE{$Q_2 \neq \emptyset$ and $w = s^{\gamma'} + t^{\delta'}$ where $(s^{\gamma'},t^{\delta'})$ is the top element of $Q_2$}
\STATE Pop $(s^{\gamma'},t^{\delta'})$. 
\IF{$\delta' + 1 \leq |V|$}
\STATE Push $(s^{\gamma'},t^{\delta'+1})$.
\ENDIF
\IF{$\mathbf{MIN_2} > w(s^{\gamma'})+w(t^{\delta'})$}
\STATE $\mathbf{MIN_2} := w(s^{\gamma'})+w(t^{\delta'})$;~ $\mathbf{SOL_2} := (\gamma',\delta')$;
\ENDIF
\ENDWHILE
\IF{$\mathbf{MIN} > \mathbf{MIN_1} + \mathbf{MIN_2}$}
\STATE $\mathbf{MIN} := \mathbf{MIN_1} + \mathbf{MIN_1}$;~ $\mathbf{SOL} := (\mathbf{SOL_1}, \mathbf{SOL_2})$;
\ENDIF
\ENDIF
\ENDWHILE
\STATE Return $\mathbf{MIN}$
\end{algorithmic}
\label{pq_algo}
\end{algorithm}

\section{Conclusions}
In this paper, we have presented $O(1.415^n)$-time and $O(1.190^n)$-space exact algorithms for 0-1 integer programs with linear equality constraints.
We can apply our algorithms to the optimization problem as well as the feasibility problem.
We can also extend our algorithms to integer programs where their variables are constrained by any finite set of integers.
There are several recent progress on the subset sum problem such as time-space tradeoff results~\cite{AKKM13} and 
improved algorithms for a certain important class of the subset sum problem~\cite{BCJ11,HG10}.
It would be interesting to investigate in these directions with connection to our results concerned with 0-1 integer programs.

Our computational experiments show that our algorithms can solve 0-1 integer programs with around 60 variables which are generated in a random way. 
Some of well-known IP solvers cannot solve these instances because they do not have any favorable structures to cut down the search space.
By connecting our algorithms to existing techniques for 0-1 integer programs (e.g., the branch-and-bound method), 
we hope that our algorithms will be useful from the practical point of view as well as theoretical analysis.


\begin{thebibliography}{10}

\bibitem{AKKM13}
P.~Austrin, P.~Kaski, M.~Koivisto, and J.~M{\"a}{\"a}tt{\"a}.
\newblock Space-time tradeoffs for subset sum: An improved worst case
  algorithm.
\newblock In {\em Proceedings of 40th International Colloquium Automata,
  Languages, and Programming}, pages 45--56, 2013.

\bibitem{BCJ11}
A.~Becker, J.-S. Coron, and A.~Joux.
\newblock Improved generic algorithms for hard knapsacks.
\newblock In {\em Proceedings of 30th Annual International Conference on the
  Theory and Applications of Cryptographic Techniques}, pages 364--385, 2011.

\bibitem{BO83}
C.~Bleich and M.~Overton.
\newblock A linear-time algorithm for the weighted median problem.
\newblock Technical Report~75, Department of Computer Science, Courant
  Institute of Mathematical Sciences, New York University, 1983.

\bibitem{Blum72}
M.~Blum, R.~W. Floyd, V.~Pratt, R.~L. Rivest, and R.~E. Tarjan.
\newblock Linear time bounds for median computations.
\newblock In {\em Proceedings of the 40th Annual ACM Symposium on Theory of
  Computing (STOC 1972)}, pages 119--124, 1972.

\bibitem{BMS05}
J.~M. Byskov, B.~A. Madsen, and B.~Skjernaa.
\newblock New algorithms for exact satisfiability.
\newblock {\em Theoretical Computer Science}, 332(1):515--541, 2005.

\bibitem{CD99}
G.~Cornu{\'e}jols and M.~Dawande.
\newblock A class of hard small 0-1 programs.
\newblock {\em INFORMS Journal on Computing}, 11(2):205--210, 1999.

\bibitem{DP02}
L.~Drori and D.~Peleg.
\newblock Faster exact solutions for some {NP}-hard problems.
\newblock {\em Theoretical Computer Science}, 287(2):473--499, 2002.

\bibitem{Fliege12}
J.~Fliege.
\newblock A randomized parallel algorithm with run time {$O(n^2)$} for solving
  an $n \times n$ system of linear equations.
\newblock {\em arXiv preprin arXiv:1209.3995}, 2012.

\bibitem{FK10}
F.~V. Fomin and D.~Kratsch.
\newblock {\em Exact exponential algorithms}.
\newblock Springer, 2010.

\bibitem{GG11}
K.~Genova and V.~Guliashki.
\newblock Linear integer programming methods and approaches -- a survey.
\newblock {\em Journal of Cybernetics and Information Technologies}, 11(1),
  2011.

\bibitem{Grover97}
L.~K. Grover.
\newblock Quantum mechanics helps in searching for a needle in a haystack.
\newblock {\em Physical Review Letters}, 79(2):325--328, July 1997.

\bibitem{HS74}
E.~Horowitz and S.~Sahni.
\newblock Computing partitions with applications to the knapsack problem.
\newblock {\em Journal of the ACM}, 21(2):277--292, 1974.

\bibitem{HG10}
N.~Howgrave-Graham and A.~Joux.
\newblock New generic algorithms for hard knapsacks.
\newblock In {\em Proceedings of 29th Annual International Conference on the
  Theory and Applications of Cryptographic Techniques}, pages 235--256, 2010.

\bibitem{ILPS14}
R.~Impagliazzo, S.~Lovett, R.~Paturi, and S.~Schneider.
\newblock 0-1 integer linear programming with a linear number of constraints.
\newblock {\em arXiv:1401.5512}, 2014.

\bibitem{IP01}
R.~Impagliazzo and R.~Paturi.
\newblock On the complexity of k-{SAT}.
\newblock {\em Journal of Computer and System Sciences}, 62(2):367--375, 2001.

\bibitem{IPS13}
R.~Impagliazzo, R.~Paturi, and S.~Schneider.
\newblock A satisfiability algorithm for sparse depth two threshold circuits.
\newblock In {\em Proceedings of the 54th Annual Symposium on the Foundations
  of Computer Science (FOCS 2013)}, pages 479--488, 2013.

\bibitem{IPZ01}
R.~Impagliazzo, R.~Paturi, and F.~Zane.
\newblock Which problems have strongly exponential complexity?
\newblock {\em Journal of Computer and System Sciences}, 63(4):512--530, 2001.

\bibitem{Matiyasevich93}
Y.~Matiyasevich.
\newblock {\em Hilbert's 10th Problem}.
\newblock The MIT Press, 1993.

\bibitem{Raghavendra12}
P.~Raghavendra.
\newblock A randomized algorithm for linear equations over prime fields.
\newblock {\em Manuscript}, 2012.

\bibitem{SS81}
R.~Schroeppel and A.~Shamir.
\newblock A {$T=O(2^{n/2})$}, {$S=O(2^{n/4}$)} algorithm for certain
  {NP}-complete problems.
\newblock {\em SIAM journal on Computing}, 10(3):456--464, 1981.

\bibitem{Williams13}
R.~Williams.
\newblock Improving exhaustive search implies superpolynomial lower bounds.
\newblock {\em SIAM Journal on Computing}, 42(3):1218--1244, 2013.

\bibitem{Woeginger03}
G.~J. Woeginger.
\newblock Exact algorithms for {NP}-hard problems: A survey.
\newblock In {\em Combinatorial Optimization -- Eureka, You Shrink!}, pages
  185--207. Springer, 2003.

\bibitem{Wolsey98}
L.~A. Wolsey.
\newblock {\em Integer Programming}.
\newblock John Wiley \& Sons, New York, 1998.

\end{thebibliography}
\end{document}